\documentclass{article}
\usepackage{amsmath}
\usepackage{amssymb}
\usepackage{amsthm}
\usepackage[titletoc]{appendix}
\usepackage{indentfirst}
\usepackage{graphicx}
\usepackage{subfigure}
\usepackage{algorithm}
\usepackage{algorithmic}
\usepackage{enumerate}
\usepackage{xcolor}
\usepackage[colorlinks, linkcolor=blue]{hyperref}
\newcommand{\abs}[1]{\lvert#1\rvert}

\newcommand{\innerp}[1]{\langle {#1} \rangle}
\newcommand{\R}{{\mathbb R}}
\newcommand{\Z}{{\mathbb Z}}

\begin{document}
\renewcommand{\algorithmicrequire}{\textbf{Inputs:}}
\renewcommand{\algorithmicensure}{\textbf{Return:}}
\newtheoremstyle{mythm}%
  {}%
  {}%
  {\itshape}%
  {\parindent}%
  {\bfseries}%
  {:}%
  {.5em}%
  {\thmname{#1}\thmnumber{ #2}\thmnote{ #3}}%
\newtheoremstyle{mythm2}%
  {}%
  {}%
  {\itshape}%
  {\parindent}%
  {\bfseries}%
  {:}%
  {.5em}%
  {\thmname{#1}\thmnumber{ #2}\thmnote{ #3}}%
\theoremstyle{mythm}
\newtheorem{theorem}{Theorem}
\theoremstyle{mythm2}
\newtheorem{lemma}{Lemma}
\theoremstyle{remark}
\newtheorem{remark}{Remark}
\newcommand{\ud}{\,\mathrm{d}} 
\newcommand{\argmin}{\mathop{\rm argmin}\limits}
\newcommand{\argmax}{\mathop{\rm argmax}\limits}
\newcommand{\sign}{{\rm sign}}
\newcommand{\diag}{{\rm diag}}
\newcommand{\supp}{{\rm supp}}

\title{ One-Bit Compressed Sensing by Greedy Algorithms\thanks{The work is supported   by NSFC grant 11171336, 11331012, 11321061. W.~Liu, D.~Gong and Z.~Xu are with Inst. Comp. Math., Academy of Mathematics and Systems Science,
Chinese Academy of Sciences, Beijing, China. Email: liuwenhui11@mails.ucas.ac.cn, gongda@lsec.cc.ac.cn, xuzq@lsec.cc.ac.cn  } }
\author{ Wenhui Liu,\,\, Da Gong,\,\, Zhiqiang Xu}

\date{\today}
\maketitle

\begin{abstract}
Sign truncated matching pursuit (STrMP) algorithm is presented in this paper. STrMP is a new greedy algorithm for the recovery of
sparse signals from the sign measurement, which combines the principle of consistent reconstruction with  orthogonal matching pursuit (OMP).
The main part of STrMP is as concise as OMP and hence STrMP is simple to implement.
In contrast to previous greedy algorithms for one-bit compressed sensing,  STrMP only need to solve a convex and unconstraint subproblem  at each iteration. Numerical experiments show that STrMP is fast and accurate for one-bit compressed sensing  compared with other algorithms.
\end{abstract}

\section{Introduction}

Compressed sensing, or compressive sensing provides a new method of data sampling and reconstruction, which allows to recover sparse signals from much fewer measurements \cite{CaRoTaSSRIIM2006, DoCS2006}. Suppose that we have an unknown sparse signal $\hat{x} \in \mathbb{R}^n$ with $\|\hat{x}\|_0 \leq s$ and $s \ll n$, where $\|\cdot\|_0$ denotes the number of nonzero components. We observe the signal as
                                         $$b = A\hat{x},$$
where $A \in \mathbb{R}^{m \times n}$  is called measurement matrix, $b \in \mathbb{R}^m$ is the vector of measurements. Compressed sensing shows that only $m = O(s\log(n/s))$ measurements are sufficient for exact reconstruction of $\hat{x}$ under many settings for the measurement matrix $A$ \cite{CaTaDLP2005, DoCS2006}.
\subsection{One-Bit Compressed Sensing}
In compressed sensing, it is supposed that the measurements have infinite bit precision. However, in practice what we get is quantized measurements.
In other words, the entries in the measurement vector $b$ must be mapped to a discrete set of values ${\mathcal A}$. There are much work
about the recovery of the general signal from the quantized measurements \cite{Wangxu}. In this paper, we  focus on the case where ${\mathcal A}=\{-1,1\}$ with the mapping being done by the sign function. So we need to recover a $s$-sparse signal from $y:=\sign(b)\in \{-1,1\}^m$. This problem is called one-bit compressed sensing, which  was first introduced by Boufounos-Baraniuk \cite{BoBa1BCS2008}.
In one-bit compressed sensing, we observe original signal as:
                                      $$y = \sign(A\hat{x}),$$
where $y \in \mathbb{R}^m$ with each element is sign of the corresponding element of $A\hat{x}$. That means we lost all magnitude information of $A\hat{x}$. Following \cite{JaLaBoBaRBCSVBSESV2013},  $x^\sharp\in \R^n$ is called  {\em a solution  for one-bit compressed sensing corresponding to $A$ and $\hat x$ } if it  satisfies
 \begin{itemize}
\item[(i)] consistence, i.e. $\sign(Ax^\sharp) =\sign(A\hat{x}) $,
\item [(ii)]sparsity, i.e. $\|x^\sharp\|_0\leq \|{\hat x}\|_0$.
\end{itemize}
A simple observation is that $ \sign(A\hat{x})=\sign(A c\hat{x})$ where $c>0$ is a scale. Thus the best one-bit compressed sensing can do is to recover $\hat x$ up to a positive scale.  Therefore, we usually expect to recover original signal on the unit Euclidean sphere in practice.

\subsection{Previous Work}
 A straightforward way to obtain a solution for one-bit compressed sensing is  to solve the following program:
										\begin{align}
										\label{opt:L0L2}
										\min~& \|x\|_0 \nonumber \\
										\text{s. t.} ~& y=\sign(Ax)~\text{and}~\|x\|_2=1.
										\end{align}
Since (\ref{opt:L0L2}) is computational intractable, similar with compressed sensing, one can replace  the $\ell_0$ norm by the more tractable $\ell_1$ norm  and obtain that ( see \cite{BoBa1BCS2008, PlVeOBCSLP2013, LaWeYiBaTBVFASRBCM2011})
   										\begin{align}
										\label{opt:L1L2}
										\min~& \|x\|_1 \nonumber \\
										\text{s. t.}~& y=\sign(Ax)~\text{and}~\|x\|_2=1.
										\end{align}

Many  algorithms have been proposed to solve (\ref{opt:L1L2}).
Particularly,  in \cite{LaWeYiBaTBVFASRBCM2011}, Laska~et.~al. use the augmented Lagrangian optimization framework to design RSS algorithm  with employing a restricted-step subroutine to solve a non-convex subproblem. Binary iterative hard thresholding (BIHT)  and adaptive outlier pursuit (AOP) are introduced in \cite{JaLaBoBaRBCSVBSESV2013} and \cite{YaYaOsRBCSUAOP2012}, respectively.
 BIHT is the modification of iterative hard thresholding which is to solve compressed sensing problem (see \cite{BlDaTHTCS2009}). AOP is a robust algorithm built on BIHT, and it is exactly BIHT when measurements are noise free. The numerical experiments in \cite{YaYaOsRBCSUAOP2012} show that AOP  performs better than the previous existing algorithms in terms of the
 recovery performance. In  \cite{PlVeOBCSLP2013}, Plan and  Vershynin replace the normalization constraint $\|x\|_2=1$ by $\|Ax\|_1=c_0$ and give an analysis of the following convex program
   										\begin{align}
										\label{opt:L1L1}
										\min~& \|x\|_1 \nonumber \\
										\text{s. t.}~& y=\sign(Ax)~\text{and}~\|Ax\|_1=c_0,
										\end{align}
where $c_0$ is a given positive constant.

 Moreover, in compressed sensing, one develops many greedy algorithms to recover the sparse signals, such as OMP, CoSaMP, ROMP, OMMP and subspace pursuit etc (see \cite{TroppOMP,TroppCOSAMP, ROMP, subspace, XuOmmp} ).
 Motivated by these algorithms in compressed sensing,  one also designs greedy algorithms for one-bit compressed sensing.
 Particulary, the matching sign pursuit (MSP) algorithm is presented in \cite{BoGSSRSM2009}.
 However, MSP suffers from time-consuming since it solves a non-convex sub-problem at  each iteration. Naturally, one may be interested in designing more efficiently greedy-type algorithm for one-bit compressed sensing, which is also the start point of this project.

 \subsection{Our Contribution}
 The aim of this paper is to present the sign truncated matching pursuit (STrMP) algorithm, which is a new greedy algorithm to solve one-bit
 compressed sensing. In particular, motivated by  \cite{PlVeOBCSLP2013}, we replace the unit $\ell_2$-norm constraint by $\|Ax\|_1=c_0$ where $c_0$ is any fixed positive constant, and hence we  consider the  following optimization problem:
                                        \begin{equation}\label{opt:L0L1}
										\begin{aligned}
										\min~& \|x\|_0  \\
										\text{s. t.}~& y=\sign(Ax)~\text{and}~\|Ax\|_1=c_0.
										\end{aligned}
                                        \end{equation}
A key step of STrMP algorithm is to use
                                          $$j_0:=\argmax_i \abs{A_i^\top y}$$
to choose the first index $j_0$. We also  prove that $j_0\in \supp({\hat x})$ with high probability provided $m=O(s\log n)$ and $A$ is a Gaussian matrix. If $j_0\in \supp({\hat x})$, we can remove the constraint $\|Ax\|_1=c_0$ and transform (\ref{opt:L0L1}) to  the program in the  form of
                                        \begin{equation*}
                                        \begin{aligned}										
										\min~& \|z\|_0  \\
										\text{\rm s. t.}~& y = \sign(Pz+q),
										\end{aligned}
                                        \end{equation*}
where $P\in \R^{m\times (n-1)}$ and $q\in \R^m$,
which is more convenient for designing greedy algorithms. We will introduce this and the algorithm in detail in section 2.  STrMP algorithm overcomes the bottleneck of MSP that a non-convex problem need to be solved at each iteration. In fact, STrMP just need to solve a convex and unconstrained sub-problem at each iteration. Hence, numerical experiments show that STrMP outperforms
previous existing  algorithms in terms of speed.  Moreover,  the numerical experiments also show that the recovery performance of STrMP is better than that of MSP and is similar with that of BIHT or AOP. The last but not the least, beside the sparsity level,  no parameter need  to be adjusted in STrMP by the user.

\subsection{Terminology and Organization}
In the following of this paper, we denote by $A_i$ the $i$th column of matrix $A$, and $A_{ij}$ the $i$th row and $j$th column component of $A$. $e_i$ denotes  the unit vector with the $i$th element is 1 and other elements are zero. For a vector $y \in \mathbb{R}^m$, $\diag(y) \in \mathbb{R}^{m \times m}$ denotes  the diagonal matrix whose diagonal elements are corresponding elements of $y$. The symbol $[n]$ denotes the index set $\{1, 2, \ldots, n\}$. For a subset $T \subset [n]$, $|T|$ denotes the number of elements in $T$. For a vector $x \in \mathbb{R}^n$ and a subset $T \subset [n]$, we use $x_T \in \mathbb{R}^{|T|}$  to   denote the vector  containing the entries of $x$ indexed by $T$. We use  $x|_T \in \mathbb{R}^n$ to denote  the vector whose entries indexed by $T$ are corresponding entries to $x$ and the entries indexed by $T^c$ are zero. For a matrix $A \in \mathbb{R}^{m \times n}$,  $A_T \in \mathbb{R}^{m \times |T|}$ denotes  a matrix which contains the columns of $A$ indexed by $T$. We define the {\em sign truncated } function as $(\cdot )_-:=\min\{\cdot,0\}$.

The rest of this paper is organized as follows. In section 2, we  derive the STrMP algorithm, and also present the  theorem which shows that one can choose the first index successfully with high probability provided $m=O(s\log n)$ and $A$ is a Gaussian matrix. The STrMP-$l_1$ algorithm, which is an adjustment of the STrMP algorithm is introduced in section 3.
The numerical results, comparing with other algorithms, are illustrated in section 4. We conclude our results in section 5.

\section{STrMP Algorithm}


We derive the algorithm STrMP in this section. The algorithm   is built  on the following theorem, which shows that one can find a  $x^\sharp$
which is a solution for one-bit compressed sensing
 by solving a program  without  the normalized  constraint
$\|Ax\|_1=c_0$.
\begin{theorem}\label{th:j0}
Suppose that $j_0\in \supp(\hat{x})$ and $y^\top A_{j_0}\neq 0$ where $y=\sign(A\hat{x})$.
Suppose that $z^\sharp\in \R^{n-1}$ is a solution to
                                        \begin{equation}\label{opt:L0L1(2)}
                                        \begin{aligned}										
										\min~& \|z\|_0  \\
										\text{\rm s. t.}~& y = \sign(Pz+q)
										\end{aligned}
                                        \end{equation}
where $P = \left(I - \frac{A_{j_0} y^\top}{y^\top A_{j_0}}\right) A_{[n]\setminus\{j_0\}}$ and $q = \frac{c_0}{y^\top A_{j_0}} A_{j_0}$.
Suppose that $x^\sharp$ is defined by
\begin{equation}\label{eq:xjing}
x^\sharp_{[n]\setminus\{j_0\}}: = z^\sharp,\quad x^\sharp_{j_0}:= \frac{c_0-y^\top A_{[n]\setminus \{j_0\}}z^\sharp}{y^\top A_{j_0}}.
\end{equation}
 Then $\|Ax^\sharp\|_1=c_0$ and $x^\sharp\in \R^n$ is  a solution  for one-bit compressed sensing corresponding to $A$ and $\hat x$.
\end{theorem}

Based on Theorem \ref{th:j0}, if  knows an index $j_0\in \supp({\hat x})$ in advance, one can construct a solution for one-bit compressed sensing corresponding to $A$ and $\hat x$ by solving a program  in the form of (\ref{opt:L0L1(2)}).
STrMP algorithm uses
                                                        $$j_0 = \argmax_{i \in [n]} |A_i^\top y|$$
to choose the index $j_0$, and we also prove that $j_0\in {\supp}({\hat x})$ with high probability provided $m=O(s\log(n-s))$ and $A$ is a Gaussian matrix:

\begin{theorem}\label{thm:Firstindex}
Let $\hat x$ be a $s$-sparse vector in $\R^n$.
Let $A \in \mathbb{R}^{m \times n}$ be a random matrix with independent standard normal entries. Assume for some
$\epsilon>0$
$$
m \geq \frac{\pi}{2} s \left( \epsilon + \sqrt{\epsilon^2 + 2\log(n-s)} \right)^2.
$$
Then
                                          $$ \argmax_{i \in [n]}{|(A^\top \sign(A\hat{x}))_i|} \in \supp(\hat{x})$$
holds with probability at least $1-2e \cdot \exp(-c\epsilon^2)$, where $c$ is an absolute constant.
\end{theorem}
We now focus on (\ref{opt:L0L1(2)}). A simple observation is that we can rewrite (\ref{opt:L0L1(2)}) as
                                        \begin{equation}\label{eq:app1}
                                        \begin{aligned}										
										\min~& \|z\|_0  \\
										\text{\rm s. t.}~&  Cz+d\geq 0,
										\end{aligned}
                                        \end{equation}
where $C=\diag(y) P$ and $d=\diag(y) q$.
We use greedy algorithms to find an approximate solution to (\ref{eq:app1}). To state the algorithm, we recall that the {\em sign truncated function } $(\cdot)_-$ which is  defined by $(\cdot)_-:=\min\{0,\cdot\}$.
The algorithm begins with an initial support set  $\Lambda^0=\emptyset$ and an estimate  $z^0=0$.
At the $k$th iteration, the algorithm computes the product of $C^\top$ and the sign truncated vector $(Cz^k+d)_-$
and form a proxy, where the truncated function $(\cdot)_-$ is applied component-wise to $Cz^k+d$. And  use
$$j^k=\argmax_i \abs{C_i^\top (Cz^k+d)_-}$$
   to choose a new index where $C_i^\top$ denotes transposition of $C_i$. Set $\Lambda^{k+1}:=\Lambda^k\cup\{j\}$.
Next we  solve a convex problem to enforce the consistence:
\begin{equation}\label{eq:bb}
{z}^{k+1}:=\argmin_z \|(Cz+d)_-\|_2^2 \qquad {\rm s.~t.}\qquad  \supp(z)\subset \Lambda_{k+1},
\end{equation}
which is essential unconstraint.
We use   two-point step size gradient method \cite{BaBoTPSSGM1988} to solve it and we introduce the algorithm in detail in Appendix B.
We conclude above fact and formulate our algorithm in Algorithm 1.
\begin{algorithm}[htb]
\caption{The Sign Truncated Matching Pursuit (STrMP) Algorithm}
\label{alg:practical}
\begin{algorithmic}[1]
\REQUIRE $A\in \R^{m\times n},~y\in \{-1,1\}^m,~c_0>0, ~s\in \Z^+$, residual tolerance $\varepsilon>0$ \par
\textbf{Initialization:}~ $\Lambda^0 = \emptyset \subset [n]$, $z^0 = 0$, $k = 0$
\STATE  Compute the first index $j_0 = \argmax_{i \in [n]} |(A^\top y)_i|$
\STATE Set
                               $$C = \diag(y) \left(I-\frac{A_{j_0}y^\top}{y^\top A_{j_0}}\right) A_{[n]\setminus \{j_0\}},~d = \frac{c_0}{y^\top A_{j_0}}\diag(y)A_{j_0}.$$

\WHILE {$\|(Cz^k+d)_-\|_2^2 \geq \varepsilon$ and $k\leq s-1$}
    \STATE {\bf Match:} $h^k=C^\top(Cz^{k}+d)_-$
	\STATE {\bf Identify:}  $j^k = \argmax_{i } |h^k_i|$
    \STATE {\bf Update:} $\Lambda^{k+1}=\Lambda^k\cup\{j^k\}$,
    \STATE \qquad\qquad $\,\,{z^{k+1}} = \argmin_{z\in \R^{n-1}} \|(Cz+d)_-\|_2^2~~{\rm s.t.}~z|_{({\Lambda}^{k+1})^c}=0$
    \STATE Increase iteration count: $k=k+1$
\ENDWHILE

\STATE Set  $x^k_{[n]\setminus\{j_0\}} = z^k,\quad x^k_{j_0}= \frac{c_0-y^\top A_{[n]\setminus \{j_0\}}z^k}{y^\top A_{j_0}}$
\ENSURE $\frac{x^k}{\|x^k\|_2}$
\end{algorithmic}
\end{algorithm}

\begin{remark}

According to Theorem \ref{th:j0}, the $x^\sharp$ defined by (\ref{eq:xjing}) is $s$-sparse and consistent. Then Theorem 2 in \cite{JaLaBoBaRBCSVBSESV2013}
shows that if $A$ is a Gaussian matrix, then
$$
\|{x^\sharp}/{\|x^\sharp\|_2}-{\hat{x}}/{\|\hat{x}\|_2}\|_2 =O\left(\frac{s}{m}\log\frac{mn}{s}\right),
$$
holds with high probability.
\end{remark}


\begin{remark}
The code in Algorithm 1 describes a version of the STrMP algorithm. Similar to OMP, there are many adjustments of STrMP. Particularly, in the identify step, one  can select many atoms per iteration instead of
one atom, which is also helpful for improving the performance.
\end{remark}

\begin{remark}
A simple observation is that the output result of STrMP algorithm is independent of $c_0$ provided $c_0\gg \varepsilon$. In fact, for different positive $c_0$,   the solution $z^{k+1}$ in the update step is the same up to a positive scale. And hence, in the identify step, STrMP algorithm chooses the same index $j$ even $c_0$ is different.
\end{remark}
\begin{remark}
Suppose that ${\hat{x}_j}=\pm \frac{1}{\sqrt{s}}$
holds for all $j\in {\rm supp}(\hat x)$. The proof of Theorem \ref{thm:Firstindex} implies that, under the
condition of Theorem \ref{thm:Firstindex},
the $s$ indices corresponding to the largest magnitude entries in the vector
 $\abs{A^\top y}$  are the support of $\hat x$ with high probability.
\end{remark}
\section{STrMP-$l_1$ Algorithm}
In the update step of the STrMP algorithm, we use $\min_z \|(Cz+d)_-\|_2$ to enforce the consistence.
 Inspired by compressed sensing, we can replace $l_2$-norm by $l_1$-norm. Thus one may consider to solve following subproblem in the update step
\begin{equation}\label{eq:STRL1}
{z}^{k+1}:=\argmin_z \|(Cz+d)_-\|_1 \qquad {\rm s.~t.}\qquad  \supp(z)\subset \Lambda_{k+1},
\end{equation}
since $l_1$-norm is more effective to characterize the sparsity than $l_2$-norm. We apply quasi-Newton method to solve the subproblem (\ref{eq:STRL1}). Correspondingly, in the match step, we compute
$$
h^k=C^\top \sign(Cz^k+d)_-,
$$
which is exactly a subgradient of $\|(Cz+d)_-\|_1$ at $z^k$ (see \cite{JaLaBoBaRBCSVBSESV2013}), and in identify step we set
$j^k=\argmax_i\abs{h_i^k}$. Here, to state conveniently, we set $\sign(0)=0$.
  In the following of this paper, we denote this algorithm by STrMP-$l_1$.



\section{Numerical Experiment}


In this section, we make numerical experiments to compare the  performance of STrMP with
that of other existing methods as mentioned before, such as BIHT, MSP and RSS with showing that both STrMP and STrMP-$l_1$ are fast and accurate.

In our experiments, the measurement matrix $A \in \mathbb{R}^{m \times n}$ is generated by Gaussian random matrix. The original signal $\hat{x} \in \mathbb{R}^n$ is $s$-sparse, and its non-zero coefficients are drawn from standard normal distribution, and it is normalized to have unit $l_2$ norm. In all following experiments, we set $n=1000$. And in all following figures, the blue line with circles denotes STrMP, the black line with upward triangle denotes STrMP-$l_1$, the yellow line with cross denotes BIHT, the red line with squares denotes RSS and the green line with hexagram denotes MSP.

\subsection{Accuracy Test}
										\begin{figure}[!htb]
										\begin{center}
										\subfigure[SNR (fixed $s=10$)]{
										\includegraphics[scale=0.32]{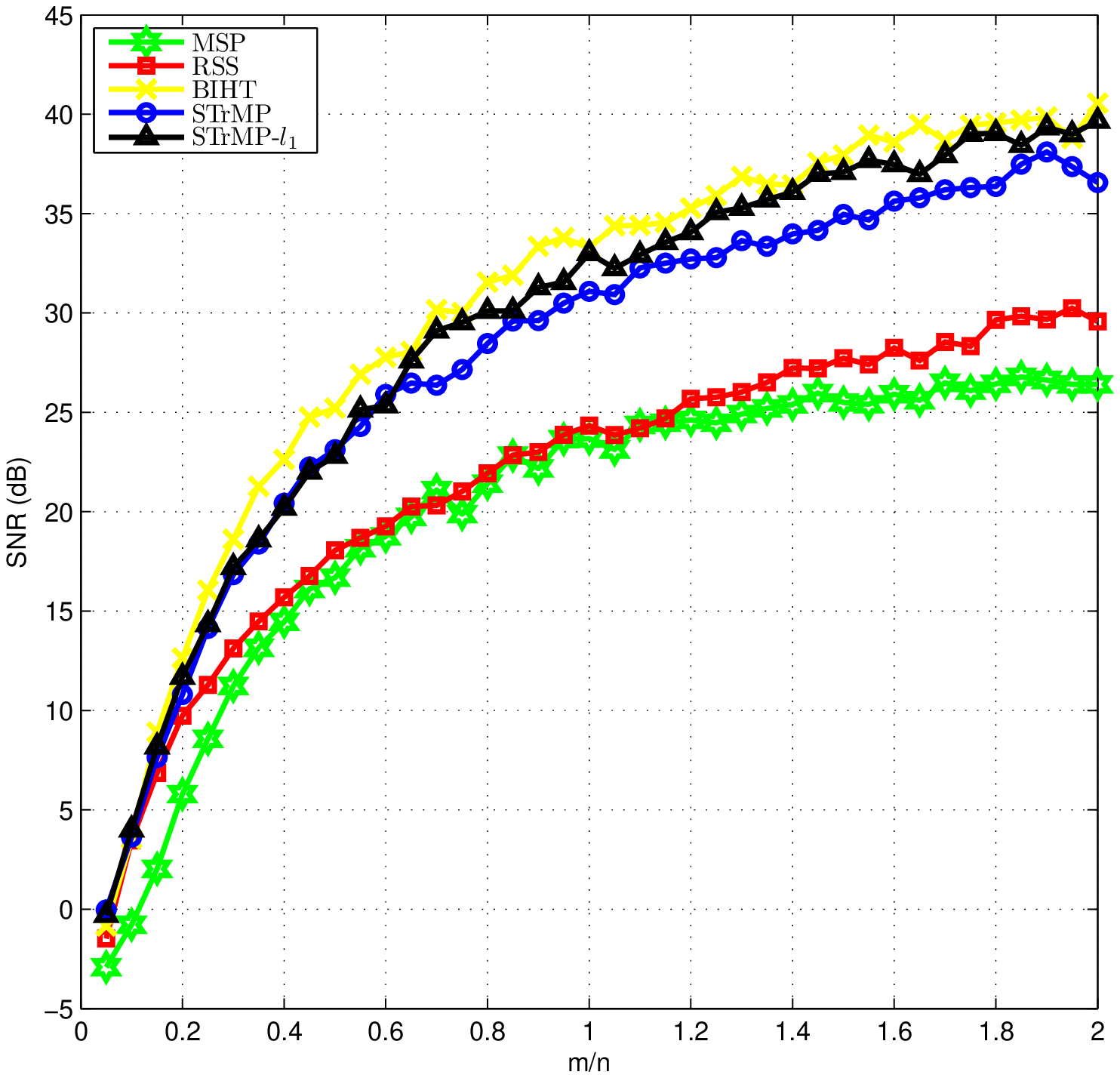}
										\label{fig:snr}} %
										\subfigure[SNR (fixed $m=1000$)]{
										\includegraphics[scale=0.32]{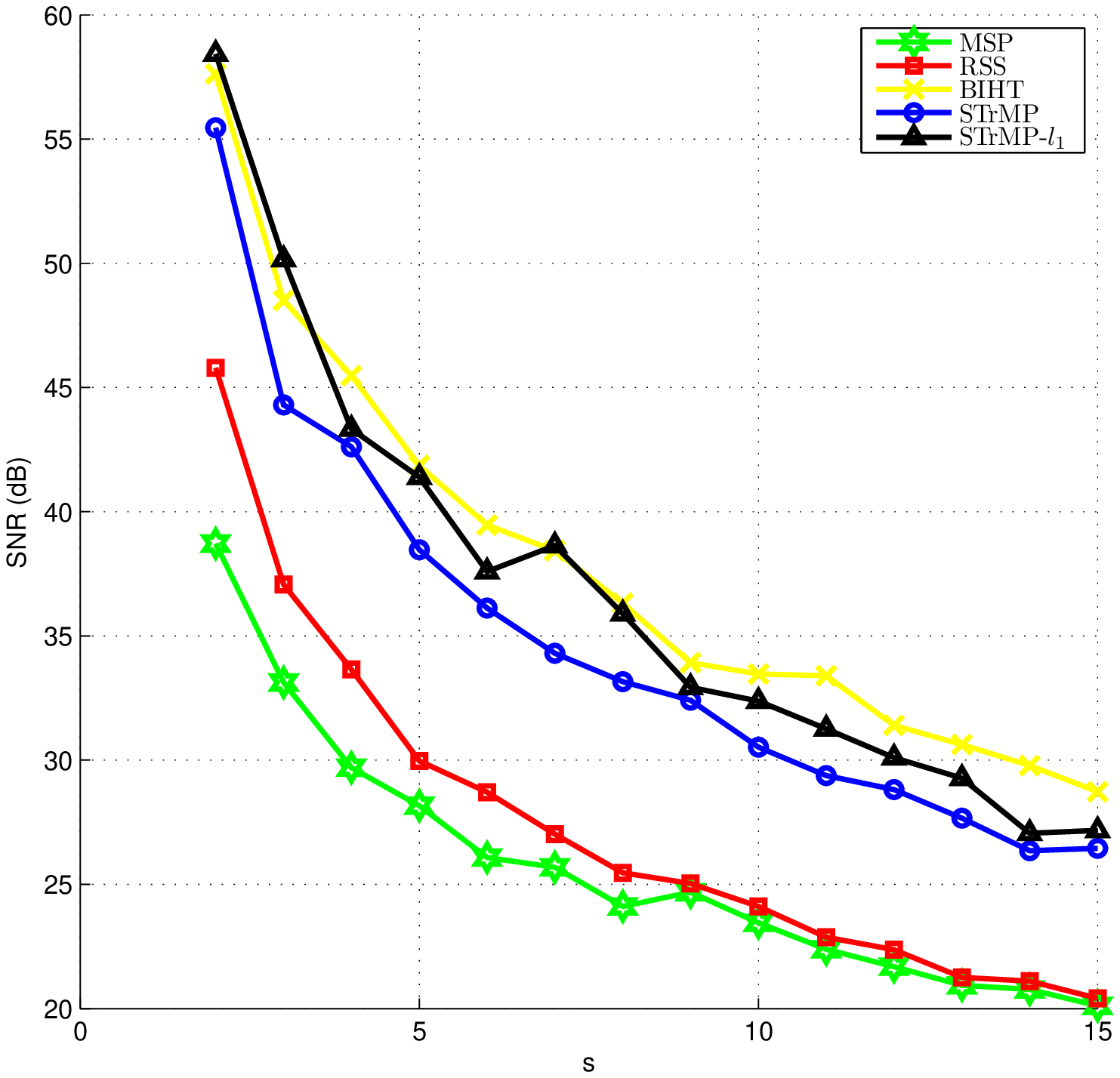}
										\label{fig:snrvss}}\\
										\subfigure[number of missed coefficients]{
										\includegraphics[scale=0.32]{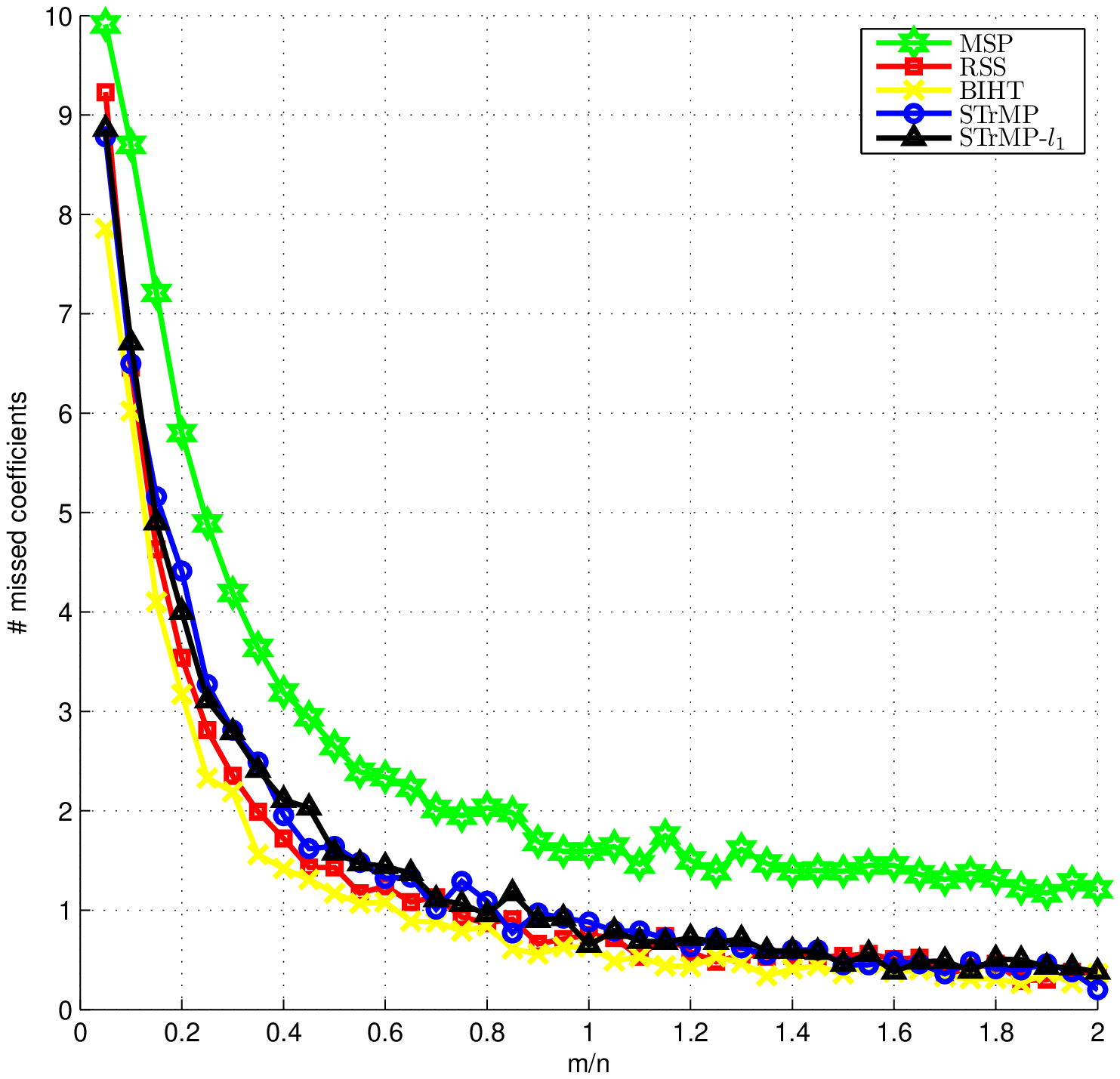}
										\label{fig:missed}}%
										\subfigure[number of misidentified coefficients]{
										\includegraphics[scale=0.32]{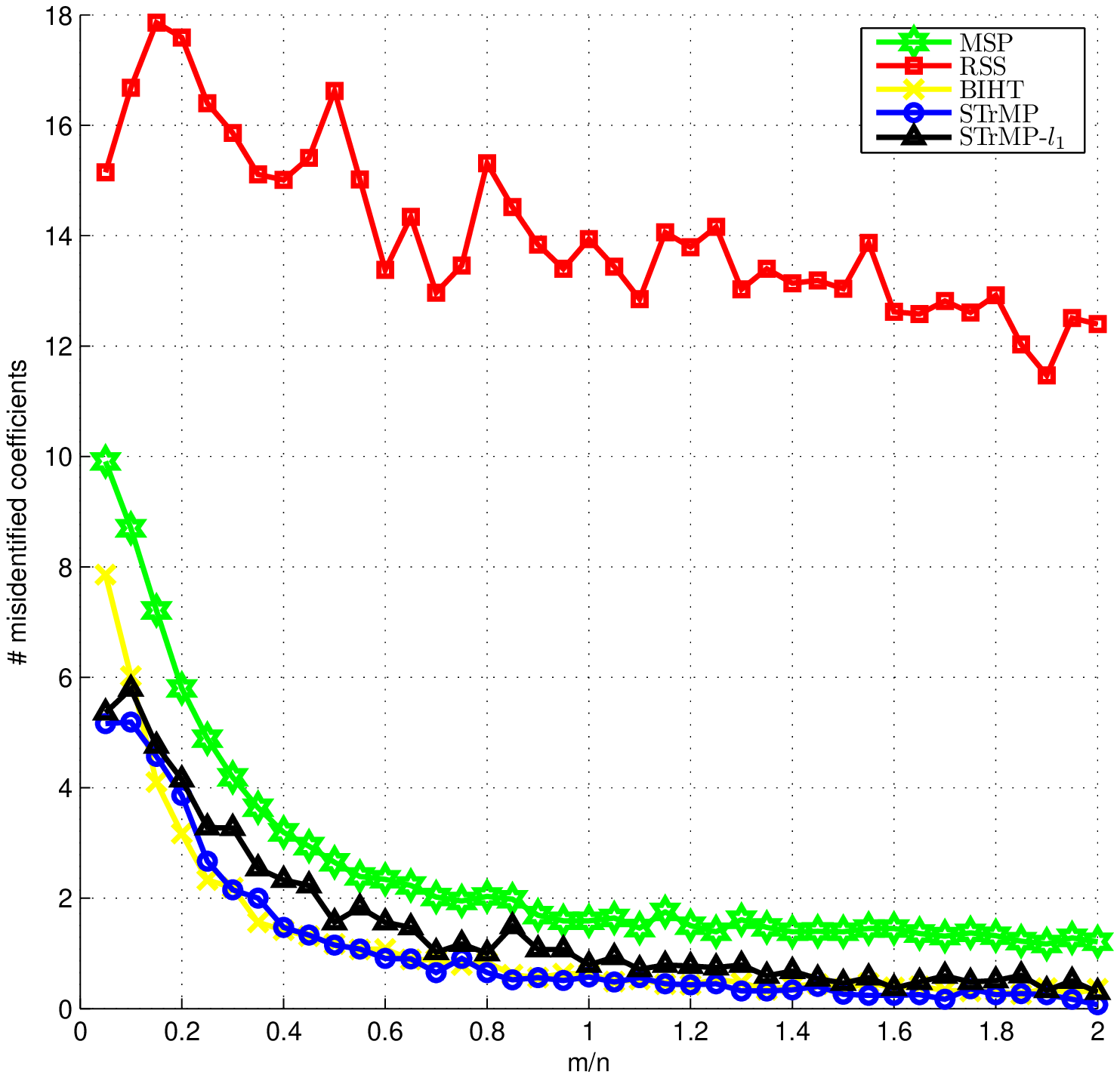}
										\label{fig:misidentified}}
										\end{center}
										\caption{Accuracy experiments: Averaged SNR  for (a) fixed $s=10$ and $n=1000$ while $m$ is changed with  $m/n$ between $0.05$ and $2$, and (b) fixed $n=1000$ and $m=1000$ while $s$ is between 1 and 15. Averaged number of missed coefficients for (c) fixed $s=10$ and $n=1000$ with $m/n$ is between $0.05$ and $2$. Averaged number of misidentified coefficients for (d) $s=10$ and $n=1000$ with $m/n$ is between 0.05 and 2. STrMP, STrMP-$l_1$ and BIHT have similar performance, RSS and MSP perform similarly. And STrMP, STrMP-$l_1$ and BIHT are better than others.}
										\label{fig:accuracy}
										\end{figure}

In this experiment, we test the reconstruction accuracy in three different ways: the average signal-to-noise ratio (SNR), the average number of missed coefficients and the average number of misidentified coefficients, which are defined as follows
\begin{itemize}
\item SNR $:= 10\log_{10}(\frac{\|x^\sharp\|^2}{\|x^\sharp-\hat{x}\|^2})$,
\item Number of missed coefficients $:= |\{i \in [n] : \hat{x}_i \neq 0~\text{and}~x_i^\sharp = 0\}|$,
\item Number of misidentified coefficients $:= |\{i \in [n] : \hat{x}_i = 0~\text{and}~x_i^\sharp \neq 0\}|$.
\end{itemize}
Here SNR is used to measure overall reconstruction performance of the recovery algorithms. Since the original signal $\hat x$ is sparse, it is also important to
identify the  support of original signal. To measure this kind of performance, it is helpful  to calculate the number of  missed coefficients and of misidentified coefficients.

The numerical results are depicted in \autoref{fig:accuracy}.
In this experiment, as mentioned before, we set $n=1000$. In  Figure \ref{fig:snr}, Figure \ref{fig:missed} and Figure \ref{fig:misidentified},
we fix $s=10$ and  change $m/n$ within the range $[0.05, 2]$ with the step  $0.05$. And hence $40$ different $m/n$ are considered.
For each $m/n$, we perform $100$ trials and record the average value. In Figure \ref{fig:snrvss}, we fix $n=1000$ and $m=1000$ while $s$ is between
$1$ and $15$. We also repeat the experiment $100$ times for each $s$ and plot the mean value. The plots in \autoref{fig:accuracy} show that STrMP, STrMP-$l_1$ and BIHT perform similarly. The performance of STrMP-$l_1$ and BIHT is slightly better than STrMP in SNR.  Figure \ref{fig:misidentified} shows that RSS exhibits poorer performance for the number of misidentified coefficients which are also observed and analyzed in \cite{LaWeYiBaTBVFASRBCM2011}.

\subsection{Consistency Test}
										\begin{figure}[!htb]
										\begin{center}
										\includegraphics[scale=0.35]{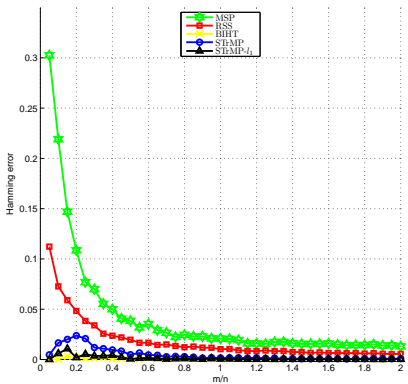}
										\caption{Consistency experiments: Averaged Hamming error between $\sign(Ax^\sharp)$ and $y$ for $s=10$ with $m/n$ is between 0.05 and 2. As the figure depicts that STrMP and STrMP-$l_1$ have the similar performance as BIHT in terms of consistency test, and they are better than others.}
										\label{fig:consistency}
										\end{center}
										\end{figure}
In this subsection, we test whether $x^\sharp$ is consistence. To do that, we  measure the Hamming error which is defined as follows:
\begin{itemize}
\item Hamming error $:= \|\sign(Ax^\sharp)-y\|_0/m$.
\end{itemize}
For each $m/n$, we also perform 100 trials and record the average values.
 The plots in \autoref{fig:consistency} show that  STrMP, STrMP-$l_1$ and BIHT work very well  and they have better performance than that of MSP and RSS.


\subsection{Speed Test}
										\begin{figure}[!htb]
										\begin{center}
										\subfigure[CPU time (fixed $s=10$)]{
										\includegraphics[scale=0.32]{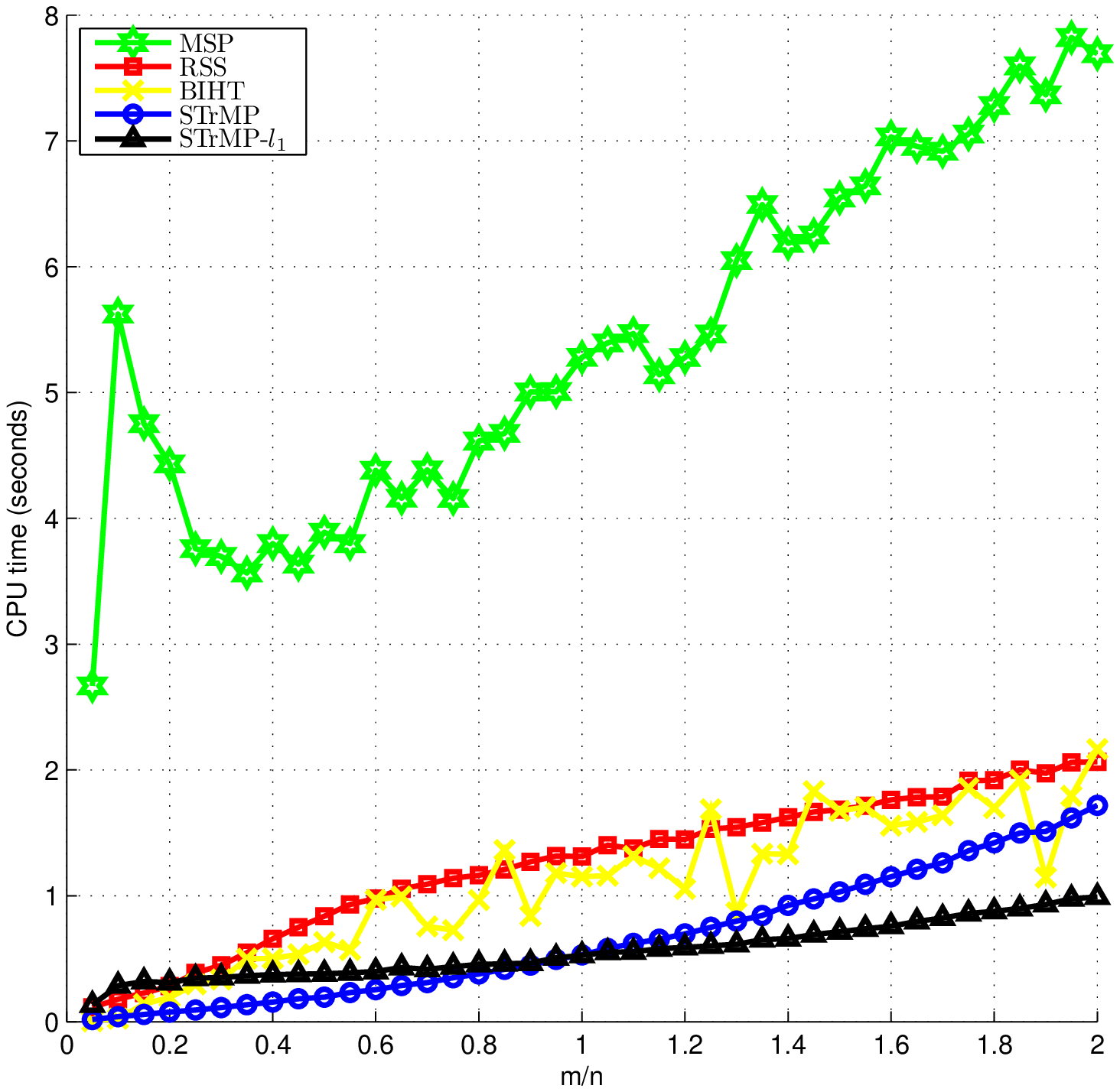}
										\label{fig:speed1}}%
										\subfigure[CPU time (fixed $m=1000$)]{
										\includegraphics[scale=0.32]{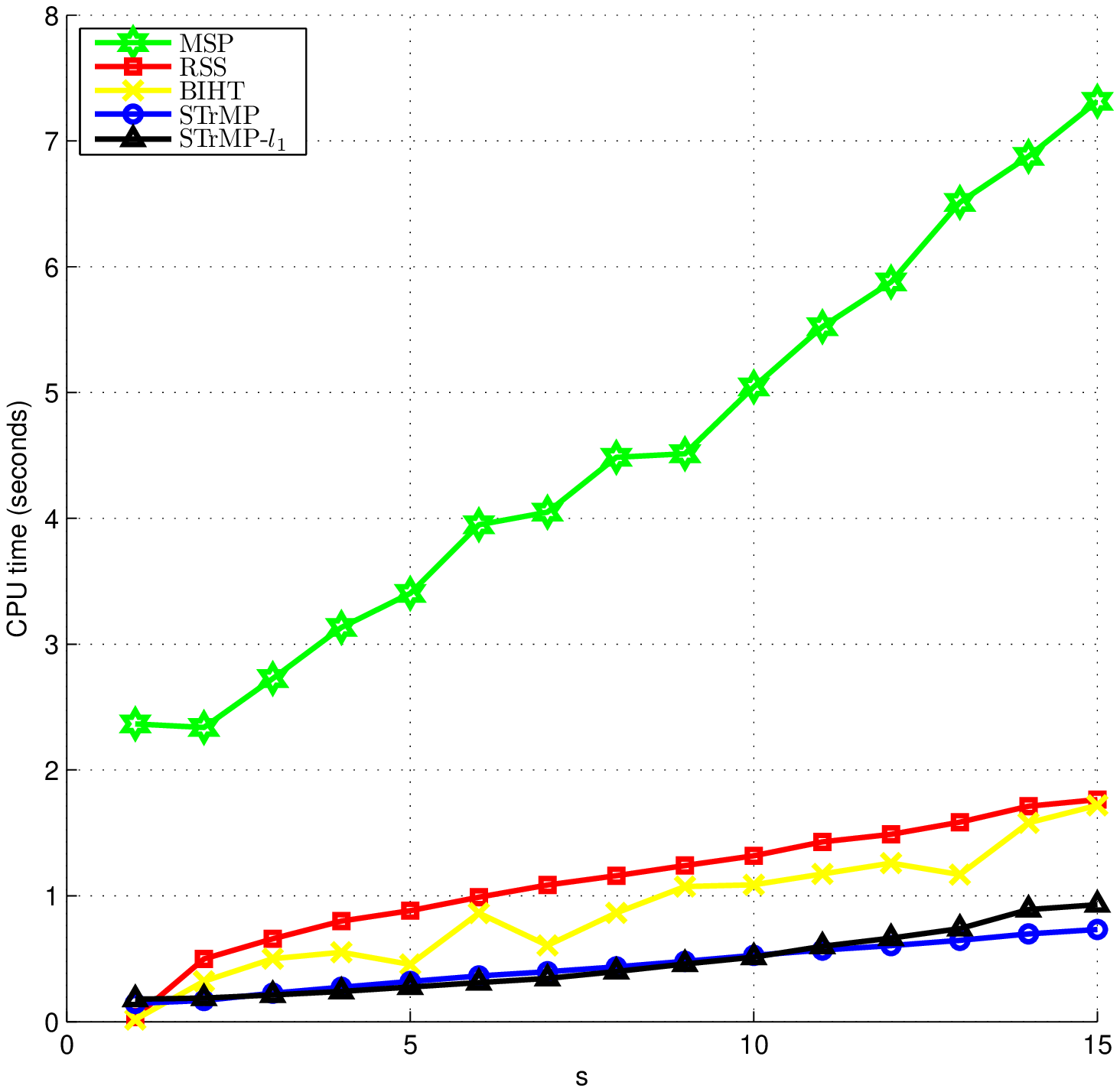}
										\label{fig:speed2}}
										\end{center}
										\caption{Speed experiments: Averaged CPU time for (a) $s=10$ with $m/n$ is between 0.05 and 2, (b) $m=1000$ with $s$ is between 1 and 15. These demonstrate that STrMP and STrMP-$l_1$ take least time comparing with other algorithms.}
										\label{fig:speed}
										\end{figure}
At last, we test the speed of the algorithms by recoding the average computational time. \autoref{fig:speed} depicts that STrMP and STrMP-$l_1$ outperform MSP significantly, and they are also faster than BIHT and RSS. The CPU time of MSP increases quickly as the measurements and the sparsity increase. However, time consuming of STrMP grows very slowly as measurements and sparsity increase, which is particularly striking for STrMP and STrMP-$l_1$.  This is because  the sub-problems of the STrMP algorithms are convex and unconstrained.


\section{Conclusion}


In this paper, we have introduced a fast and accurate greedy algorithm for  the one-bit compressive sensing.
The subproblem in STrMP is convex and unconstraint. And hence, STrMP algorithm is
faster than  previously existing one-bit compressed sensing algorithms. The choose of the first index plays an important role in STrMP.  We  prove that  the first index belongs to the support of the original signal
with high probability provided $m=O(s\log n)$.  The numerical experiments show that the recovery performance of STrMP and STrMP-$l_1$ algorithm are similar with that of BIHT, and they are much better than that of RSS and MSP. One can note that the main part of STrMP is as concise as OMP. So, it will be very interesting to investigate the convergence property of STrMP using the technology developed in the study of OMP \cite{TroppOMP, TroppCOSAMP, ZhSROMPRIP2011}.

 \begin{appendices}

    \renewcommand{\thesection}{{\bf Appendix \Alph{section}}}

\section{Proofs of Theorems 1 and  2}


\begin{proof}[Proof of Theorem \ref{th:j0}]
To this end, we only need to prove that $x^\sharp$ satisfies
 \begin{itemize}
\item[(i)] consistence, i.e. $\sign(Ax^\sharp) =y $,
\item [(ii)]sparsity, i.e. $\|x^\sharp\|_0\leq \|{\hat x}\|_0$,
\item [(iii)] normalization, i.e. $\|Ax^\sharp\|_1=c_0$.
\end{itemize}
We first consider
\begin{eqnarray*}
\sign(Ax^\sharp) &=& \sign(A_{[n]\setminus\{j_0\}}z^\sharp+x_{j_0}^\sharp A_{j_0})\\
 &=& \sign\left(\left(I-\frac{A_{j_0} y^\top}{y^\top A_{j_0}}\right) A_{[n]\setminus \{j_0\}}z^\sharp + \frac{c_0}{y^\top A_{j_0}}A_{j_0}\right)\\
 &=& \sign(Pz^\sharp+q)=y,
\end{eqnarray*}
which implies that (i).
We turn to (iii).
Note that
\begin{eqnarray*}
     \|Ax^\sharp\|_1&=&y^\top Ax^\sharp = y^\top A_{[n]\setminus j_0}z^\sharp + y^\top A_{j_0}x^\sharp_{j_0}\\
      &=& y^\top A_{[n]\setminus j_0} z^\sharp + y^\top A_{j_0} \frac{c_0-y^\top A_{[n]\setminus j_0}z^\sharp}{y^\top A_{j_0}} =  c_0.
\end{eqnarray*}
We arrive at (iii).
We next consider (ii).
To this end, we need to show that
$$
\|x^\sharp\|_0\leq \|\hat{x}\|_0.
$$
Without loss of generality, we can assume that $\|A\hat{x}\|_1=y^\top A\hat{x}=c_0$ (otherwise, we can multiply $\hat x$ by a positive constant  ),
which implies that
\begin{equation}\label{eq:xj0}
\hat{x}_{j_0}= \frac{c_0-y^\top A_{[n]\setminus \{j_0\}}\hat{x}_{[n]\setminus j_0}}{y^\top A_{j_0}}.
\end{equation}
Then a simple calculation shows that
\begin{eqnarray*}
y&=&\sign(A{\hat x})= \sign(A_{[n]\setminus j_0}{\hat x}_{[n]\setminus j_0}+A_{j_0}{\hat x}_{j_0})\\
&=&\sign(P\hat{x}_{[n]\setminus \{j_0\}}+q).
\end{eqnarray*}
Here, in the last equality, we use (\ref{eq:xj0}).
And hence $\hat{x}_{[n]\setminus \{j_0\}}$ satisfies the constraint condition  (\ref{opt:L0L1(2)}) which implies that
$$
\|z^\sharp\|_0 \leq  \|\hat{x}_{[n]\setminus \{j_0\}}\|_0=\|\hat{x}\|_0-1.
$$
So, the definition of $x^\sharp$ implies that
$$
\|x^\sharp\|_0 \leq \|z^\sharp\|_0+1 \leq \|\hat{x}\|_0.
$$
We arrive at the conclusion.
\end{proof}

To prove  Theorem \ref{thm:Firstindex}, we first introduce Hoeffding-type inequality  (see \cite[Propsition~5.10]{VeITNAARM2012}).



\begin{lemma}[(Hoeffding-type inequality)]
Let $\zeta_1,\ldots, \zeta_n$ be independent centered sub-gaussian random variables, and $K = \max\limits_i \|\zeta_i\|_{\psi_2}$. Then for every $a=(a_1,\ldots, a_n) \in \mathbb{R}^n$ and every $\epsilon \geq 0$, we have
              $$P\{|\sum\limits_{i=1}^{n} a_i \zeta_i| \geq \epsilon\} \leq e \cdot \exp(-\frac{c_1\epsilon^2}{K^2 \|a\|_2^2}),$$
where $c_1 > 0$ is an absolute constant.
\end{lemma}

 Here, for a sub-gaussian random variable $\xi$,
 $$
 \|\xi\|_{\psi_2}:=\mathop{\rm sup}\limits_{p\geq 1}({\mathbb E}\abs{\xi}^p)^{1/p}
 $$
  (see also \cite[Definition~5.7]{VeITNAARM2012}). For a standard normal variable $\xi$, $\|\xi\|_{\psi_2}$ is bounded by $\sqrt{2}$. We are now turning to the proof of our theorem.

\begin{proof}[Proof of Theorem \ref{thm:Firstindex}]
Without loss of generality, we can suppose that $\|\hat{x}\|_2 = 1$
and  $\hat{x}_\jmath>0$ where $\jmath := \argmax_{i \in [n]}|\hat{x}_i|$.
We claim that, for any $\epsilon > 0$
\begin{align}\label{eq:claim1}
{\mathbb P}\{ \max\limits_{k \notin \supp(\hat{x})} \abs{\innerp{A_k,y}} < \sqrt{m\epsilon^2 + 2m\log(n-s)} \} \geq 1 - \exp(-\frac{\epsilon^2}{2}),
\end{align}
and
\begin{align}\label{eq:claim2}
  {\mathbb P}\{|\innerp{A_\jmath,y}| \geq 2\cdot m\frac{\hat{x}_\jmath}{\sqrt{2\pi}}- \epsilon\sqrt{m}\}  \geq 1 - e \cdot \exp(-\frac{c_1\epsilon^2}{8}).
\end{align}
Combining (\ref{eq:claim1}) and (\ref{eq:claim2}), to prove  the conclusion we just need to show that
\begin{eqnarray}\label{eq:in>notin}
2\cdot m\frac{\hat{x}_\jmath}{\sqrt{2\pi}} - \epsilon \sqrt{m} \geq \sqrt{m\epsilon^2 + 2m\log(n-s)},
\end{eqnarray}
which is equivalent to
\begin{equation}\label{eq:in>notin1}
\sqrt{m}\cdot \frac{2\hat{x}_\jmath}{\sqrt{2\pi}}  \geq \sqrt{\epsilon^2 + 2\log(n-s)}+ \epsilon .
\end{equation}
Indeed, note that $\hat{x}_\jmath = \|\hat{x}\|_\infty\geq \frac{1}{\sqrt{s}}$ with $\hat{x}$ being $s$-sparse, we have
\begin{eqnarray*}
\sqrt{m}\cdot\frac{2\hat{x}_\jmath}{\sqrt{2\pi}} \geq \sqrt{\frac{2}{{\pi}}} \cdot \sqrt{\frac{m}{s}} \geq  \sqrt{\epsilon^2 + 2\log(n-s)}+ \epsilon,
\end{eqnarray*}
which implies  (\ref{eq:in>notin1}) and hence (\ref{eq:in>notin}) holds.
Here, in the last inequality, we use
\begin{eqnarray}\label{eq:condofm}
m \geq \frac{\pi}{2} s ( \epsilon + \sqrt{\epsilon^2 + 2\log(n-s)} )^2.
\end{eqnarray}
Hence, if the condition (\ref{eq:condofm}) is satisfied, we have
                                          $$ \argmax_{i \in [n]}{|(A^\top \sign(A\hat{x}))_i|} \in \supp(\hat{x})$$
holds, with probability at least $1-2e \cdot \exp(-c\epsilon^2)$, which implies the conclusion,
where $c:=\min\{1,c_1/8\}$.

 To this end, we still need to prove (\ref{eq:claim1}) and (\ref{eq:claim2}).
 We first consider (\ref{eq:claim1}). According to $y = \sign(A\hat{x})$, we obtain that the entries of $y$ are $i.i.d.$ Bernoulli random variables, and they are independent of entries of $A_k$ for those $k \notin \supp(\hat{x})$.
 Note that $\innerp{A_k,y} \sim \mathcal{N}(0,m)$ provided $k \notin \supp(\hat{x})$. The Gaussian concentration inequality implies
                       $$P\{|\innerp{A_k,y}| <  t\} \geq 1 - \exp\left(-\frac{t^2}{2m}\right),$$
 where $t>0$.
Since $\hat{x}$ is $s$-sparse, there are $(n-s)$ entries do not belong to $\supp(\hat{x})$. By using  the union bound, we have
           $$
           {\mathbb P}\{ \max\limits_{k \notin \supp(\hat{x})} |\innerp{A_k,y}| < t\} \geq 1 - (n-s)\exp\left(-\frac{t^2}{2m}\right).
           $$
Taking $t=\sqrt{m\epsilon^2 + 2m\log(n-s)}$, we obtain that
   $${\mathbb P}\{ \max\limits_{k \notin \supp(\hat{x})} |\innerp{A_k,y}| < \sqrt{m\epsilon^2 + 2m\log(n-s)} \} \geq 1 - \exp\left(-\frac{\epsilon^2}{2}\right).$$

We next turn to (\ref{eq:claim2}). Without loss of generality we can assume that $\hat{x}_\jmath > 0$.
Note that $A_{i\jmath}y_i,~i \in [m]$ are $i.i.d.$ random variables since $y_i$ only depends on the $i$th row of $A$.  Then for every $\epsilon \geq 0$
$${\mathbb P}\{|A_{i\jmath}y_i| \geq \epsilon\} ={\mathbb P}\{|A_{i\jmath}| \geq \epsilon\} \leq \exp(-\frac{\epsilon^2}{2}),$$
which implies that  $A_{i\jmath}y_i,~i \in [m]$ are $i.i.d.$ sub-gaussian random variables.
 We next  calculate  the expectation of $A_{i\jmath}y_i$.
Let $\xi := A_{i\jmath}$ and $\eta := \sum\limits_{k \neq \jmath} A_{ik} \hat{x}_k$. Then $\xi \sim \mathcal{N}(0,1)$ and $\eta \sim \mathcal{N}(0,1-\hat{x}_\jmath^2)$. Also, note that $\xi$ and $\eta$ are independent
and
$${\mathbb E}(A_{i\jmath}y_i) ={\mathbb E}( A_{i\jmath}\cdot \sign( A_{i\jmath}\hat{x}_\jmath + \sum\limits_{k \neq \jmath} A_{ik} \hat{x}_k ) )={\mathbb E} (\xi\cdot \sign(\xi \hat{x}_\jmath+ \eta)).$$
 Now we have
\allowdisplaybreaks
\begin{eqnarray}
{\mathbb E}(A_{i\jmath}y_i)
& = & \int_{-\infty}^{\infty} \int_{-\infty}^{\infty} \xi\cdot \sign(\xi \hat{x}_\jmath + \eta) \frac{1}{\sqrt{2\pi}} e^{-\frac{\xi^2}{2}} \frac{1}{\sqrt{2\pi(1-\hat{x}_\jmath^2)}} e^{-\frac{\eta^2}{2(1-\hat{x}_\jmath^2)}} \ud \xi \ud \eta \nonumber \\
& = & \int_{-\infty}^{\infty} \int_{-\infty}^{-\frac{\eta}{\hat{x}_\jmath}} (-\xi) \frac{1}{\sqrt{2\pi}} e^{-\frac{\xi^2}{2}} \frac{1}{\sqrt{2\pi(1-\hat{x}_\jmath^2)}} e^{-\frac{\eta^2}{2(1-\hat{x}_\jmath^2)}} \ud \xi \ud \eta \nonumber \\
&   & +\: \int_{-\infty}^{\infty} \int_{-\frac{\eta}{\hat{x}_\jmath}}^{\infty} \xi \frac{1}{\sqrt{2\pi}} e^{-\frac{\xi^2}{2}} \frac{1}{\sqrt{2\pi(1-\hat{x}_\jmath^2)}} e^{-\frac{\eta^2}{2(1-\hat{x}_\jmath^2)}} \ud \xi \ud \eta \\
&=& 2\frac{\hat{x}_\jmath}{\sqrt{2\pi}}\label{eq:twointsum}.
\end{eqnarray}
By using Lemma 1, we can get
    $${\mathbb P}\{|\sum\limits_{i=1}^{m} (A_{i\jmath}y_i - {\mathbb E}(A_{i\jmath}y_i))| \geq \epsilon  \sqrt{m}\}  \leq e \cdot \exp(-{c_1\epsilon^2/8}),$$
where $c_1$ is an absolute constant. Then
  $${\mathbb P}\{|\sum\limits_{i=1}^{m} A_{i\jmath}y_i| \geq m\cdot {\mathbb E}(A_{i\jmath}y_i) - \epsilon\sqrt{m}\}  \geq 1 - e \cdot \exp(-c_1\epsilon^2/8).$$
\end{proof}

\begin{remark}
We also make numerical experiments to test the success probability of the choose of the first index.
 We set $n=1000, s=15$ and change $m$ within the range $[30,200]$.
  We repeat the experiment $100$ times for each $m$ and calculate the success rate.
 The numerical results show that  $m \thickapprox 150$ measurements are enough for correctly selecting the first index with the success rate being 1.
\end{remark}

\section{Two-Point Step Size Gradient Method}
The Step 6 of STrMP algorithms is to use the two-point step size gradient method \cite{BaBoTPSSGM1988}. We show it in  Algorithm \ref{alg:bbalg}. The algorithm is usually to solve unconstrained optimization problems in the form of
\begin{equation}
\min_{x \in {\R}^n} f(x),
\end{equation}
where $f(x)$ is derivable. For more details see \cite{BaBoTPSSGM1988}.
\begin{algorithm}[htb]
\caption{Two-Point Step Size Gradient Method}
\label{alg:bbalg}
\begin{algorithmic}[1]
\REQUIRE $ 0 \leq \varepsilon \ll 1 $ \par
\textbf{Initialization:}~$x_1 \in {\R}^n$,~$k:=1$
\WHILE {$\|\nabla f(x_k)\|_2 > \varepsilon$}
	\STATE $d_k = -\nabla f(x_k)$
	\IF {$k=1$}
	\STATE using exact line search or other methods to find first step $\alpha_1$.
	\ELSE
	\STATE $\alpha_k = s_{k-1}^\top y_{k-1} / \|y_{k-1}\|_2^2$, where $$s_{k-1} = x_k - x_{k-1},~y_{k-1} = \nabla f(x_k)-\nabla f(x_{k-1}).$$
	\ENDIF
	\STATE $x_{k+1} := x_k + \alpha_k d_k$, $k:=k+1$
\ENDWHILE
\ENSURE $x_k$
\end{algorithmic}
\end{algorithm}

\end{appendices}

\bibliographystyle{plain}

\end{document}